\documentclass[a4paper, 11pt]{article}

\usepackage{epsfig}
\usepackage{amsmath}
\usepackage{amssymb}
\usepackage{latexsym}
\usepackage{amsfonts}
\usepackage{amsthm}
\usepackage[numbers,sort&compress]{natbib}
\usepackage{graphicx}
\ifx\pdfoutput\undefined  \else
 \fi
\usepackage[centerlast]{caption2}
\usepackage{color}

\newtheorem{lemma}{Lemma}
\newtheorem{theorem}{Theorem}
\newtheorem{corollary}{Corollary}

\newtheorem{conjecture}{Conjecture}
\numberwithin{figure}{section} \numberwithin{definition}{section}
\numberwithin{observation}{section} \numberwithin{lemma}{section}\numberwithin{corollary}{section}
\numberwithin{theorem}{section} \numberwithin{proposition}{section}
\numberwithin{conjecture}{section} \numberwithin{table}{section}

\begin{document}
\title{
{The crossing numbers of $K_{n,n}-nK_2$, $K_{n}\times P_2$, $K_{n}\times P_3$ and $K_n\times C_4$}\footnote{The research is supported by NSFC of China (No.60973014,\ 61170303)}}
\author{
Yuansheng Yang \footnote {corresponding
author's email : yangys@dlut.edu.cn}, \ Baigong Zheng, \ Xiaohui Lin, \ Xirong Xu \\
School of Computer Science and Technology \\
Dalian University of Technology, Dalian, 116024, P.R. China }

\date{}
\maketitle
\begin{abstract}
The {\it crossing number} of a graph $G$ is the minimum number of
pairwise intersections of edges among all drawings of $G$. In this
paper, we study the crossing number of $K_{n,n}-nK_2$, $K_n\times
P_2$, $K_n\times P_3$ and $K_n\times C_4$.
\bigskip

\end{abstract}

\section{Introduction and preliminaries}

\indent \indent Let $G$ be a graph, $V(G)$ the vertex set and $E(G)$
the edge set. The {\it crossing number} of $G$, denoted by $cr(G)$,
is the smallest number of pairwise crossings of edges among all
drawings of $G$ in the plane. We use $D$ to denote a drawing of a
graph $G$ and $\nu(D)$ the number of crossings in $D$. It is clear
that $cr(G)\leq \nu(D)$.

Let $E_1$ and $E_2$ be two disjoint subsets of an edge set $E$. The
number of the crossings formed by an edge in $E_1$ and another edge
in $E_2$ is denoted by $\nu_D(E_1,E_2)$ in a drawing $D$. The number
of the crossings that involve a pair of edges in $E_1$ is denoted by
$\nu_D(E_1)$. Then $\nu(D)=\nu_D(E(G))$ and $\nu_D(E_1\cup
E_2)=\nu_D(E_1)+\nu_D(E_2)+\nu_D(E_1,E_2)$.

The $Kronecker$ $product$ $G\times H$ of graphs $G$ and $H$ has
vertex set $V(G\times H)=V(G)\times V(H)$ and edge set $E(G\times
H)=\{\{(a,x),(b,y)\}:\{a,b\}\in E(G)$ and $\{x,y\}\in E(H)\}$. (The
product is also known as direct product, cardinal product, cross
product and graph conjunction.)

Computing the crossing number of graphs is a complicated yet
classical problem. It is proved that the problem is NP-complete by
Garey and Johnson \cite{GJ83}.

Zarankiewicz \cite{Z54} gave a drawing of $K_{m,n}$, which
demonstrated that
$$cr(K_{m,n})\leq Z(m,n)=\lfloor\frac{m}{2}\rfloor\lfloor\frac{m-1}{2}\rfloor\lfloor\frac{n}{2}\rfloor\lfloor\frac{n-1}{2}\rfloor.$$
The equality holds for $min(m,n)\leq 6$ \cite{K70} and $7\leq
n\leq8$, $7\leq m\leq10$ \cite{W93}.

Gue \cite{G69} and Bla$\check{z}$ek and Koman \cite{BK64} gave a
drawing of $K_{n}$ in a cylinder (homeomorphic to a plane, a
sphere), which demonstrated that
$$cr(K_{n})\leq Z(n)=\frac{1}{4}\lfloor\frac{n}{2}\rfloor\lfloor\frac{n-1}{2}\rfloor\lfloor\frac{n-2}{2}\rfloor\lfloor\frac{n-3}{2}\rfloor.$$
The equality holds for $n\leq12$ \cite{A96}.

De Klerk, Pasechnik and Schrijver \cite{D07} proved $cr(K_{m,n})\geq 0.8594Z(m,n)$ and $cr(K_{n})\geq 0.8594Z(n)$.

In literature, the Cartesian product has been paid more attention
\cite{A04,Z07,Z08,Z081}, while Kronecker product has fewer results
on the crossing number \cite{JD12}.

In this paper, we study the crossing numbers of $K_{n,n}-nK_2$ and
the Kronecker product $K_n\times P_2$, $K_n\times P_3$, $K_n\times
C_4$. In Section 2, we give upper bounds of $cr(K_{n,n}-nK_2)$ and
$cr(K_{n}\times P_2)$ by constructing a drawing of $K_{n,n}-nK_2$ in
a cylinder. In Section 3, we give upper bound of $cr(K_{n}\times
P_3)$ by constructing a drawing of $K_n\times P_3$ based on the
drawing of $K_{n,n}-nK_2$. In Section 4, we give upper bound of
$cr(K_{n}\times C_4)$ by constructing a drawing of $K_n\times C_4$
based on the drawing of $K_{n,n}-nK_2$, too. In Section 5, we give
lower bounds of $cr(K_{n,n}-nK_2)$, $cr(K_{n}\times P_2)$,
$cr(K_{n}\times P_3)$ and $cr(K_{n}\times C_4)$.

\section{Upper bounds of $cr(K_{n,n}-nK_2)$ and $cr(K_{n}\times P_2)$}

\indent \indent Let
$$\begin{array}{llll}
V(K_{n,n}-nK_2)&=\{a_{i},b_{i}\ |\ 0\leq i\leq n-1\},\\
E(K_{n,n}-nK_2)&=\{(a_{i},b_{j})\ |\ 0\leq i\neq j\leq n-1\}.
\end{array}$$

In Figure  \ref{fig: Dn}, we exhibit drawings $D_n$ of
$K_{n,n}-nK_2$ in a cylinder for $n\leq 13$. A cylinder can be
`assembled' from a polygon by identifying one pair of opposite sides
of a rectangle \cite{BW78}.

\begin{figure}
\centering
\includegraphics[scale=1.0]{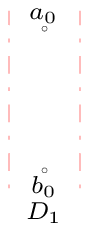}
\hspace{4bp}
\includegraphics[scale=1.0]{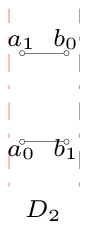}
\hspace{4bp}
\includegraphics[scale=1.0]{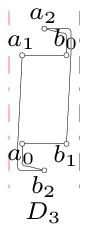}
\hspace{4bp}
\includegraphics[scale=1.0]{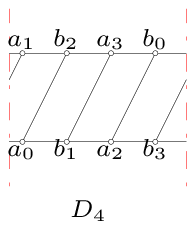}
\hspace{4bp}
\includegraphics[scale=1.0]{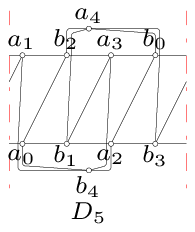}
\hspace{4bp}
\includegraphics[scale=1.0]{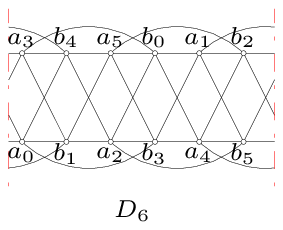}
\hspace{4bp}
\includegraphics[scale=1.0]{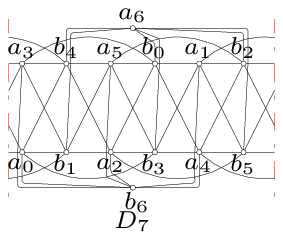}
\includegraphics[scale=1.0]{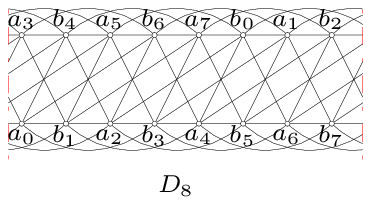}
\includegraphics[scale=1.0]{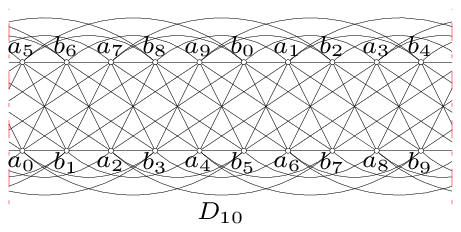}
\includegraphics[scale=1.0]{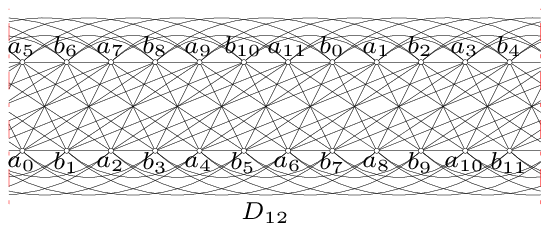}
\includegraphics[scale=1.0]{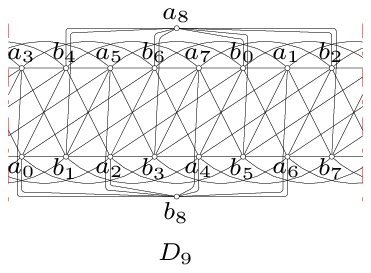}
\includegraphics[scale=1.0]{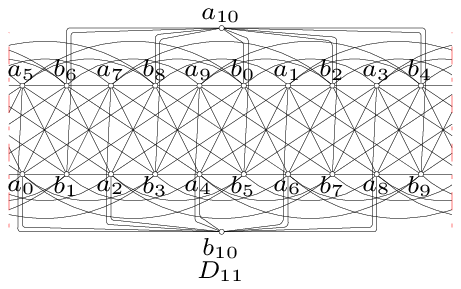}
\includegraphics[scale=1.0]{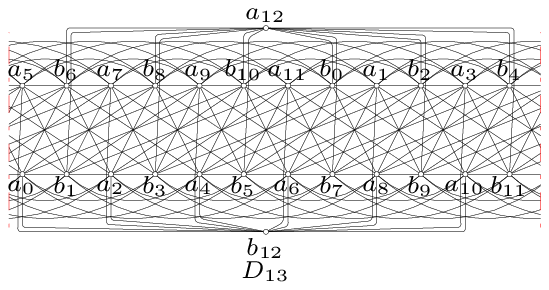}
\caption{\small{Drawings $D_n$ for $n\leq 13$}}\label{fig: Dn}
\end{figure}

From drawings in Figure \ref{fig: Dn}, we have
\begin{lemma}\label{lemma: D1}
For $n\leq 4$, $cr(K_{n,n}-nK_2)=0$.
\end{lemma}

For even $n\geq 6$, let
$$\begin{array}{llll}
X&=\{a_{2i},b_{2i+1} \ |\ 0\leq i\leq n/2-1\},\\
Y&=\{a_{2i+1},b_{2i} \ |\ 0\leq i\leq n/2-1\},\\
E_X&=\{uv \ |\ u,v\in X\},\\
E_{XY}&=\{uv \ |\ u\in X \mbox{ and } v\in Y\},\\
E_Y&=\{uv \ |\ u,v\in Y\}.
\end{array}$$

Then
$$\begin{array}{llll}
V(K_{n,n}-nK_2)&=X\bigcup Y,\ \  X\bigcap Y=\emptyset.\\
E(K_{n,n}-nK_2)&=E_X\bigcup E_{XY}\bigcup E_{Y}.
\end{array}$$

In drawings $D_{n}$ with even $n\geq 6$, vertices
$a_{2\lceil\frac{n}{4}\rceil-1},b_{2\lceil\frac{n}{4}\rceil},\cdots,a_{2\lceil\frac{n}{4}\rceil-3},b_{2\lceil\frac{n}{4}\rceil-2}$
($a_{0},$ $b_{1},\cdots,a_{n-2},b_{n-1}$) of $X$ $(Y)$ are placed
equidistantly on the perimeter of the top (bottom) disk, edges of
$E_X$ $(E_Y)$ are drawn on the top (bottom) disk, and edges of
$E_{XY}$ are drawn by shortest helical curves on the cylinder. It
follows
$\nu_{D_{n}}(E_X,E_{XY})=\nu_{D_{n}}(E_Y,E_{XY})=\nu_{D_{n}}(E_X,E_{Y})=0$
and
$\nu(D_{n})=\nu_{D_{n}}(E_X)+\nu_{D_{n}}(E_Y)+\nu_{D_{n}}(E_{XY})=2\nu_{D_{n}}(E_X)+\nu_{D_{n}}(E_{XY})$.

\noindent For even $n\geq 6$ and $n$ mod 4=2,
$$\begin{array}{llll}
\nu_{D_{n}}(E_X)&=\frac{1}{2}(n\sum^{\frac{n-6}{4}}_{d=1}2d(\frac{n}{2}-1-d)+\frac{n}{2}2(\frac{n-2}{4})^2)\\
&=\frac{n(n-2)}{2}\frac{\frac{n-6}{4}\frac{n-2}{4}}{2}-n\frac{\frac{n-6}{4}\frac{n-2}{4}\frac{n-4}{2}}{6}+\frac{n(n-2)^2}{32}\\
&=\frac{n(n-2)(3n^2-24n+36-n^2+10n-24+6n-12)}{192}\\
&=\frac{n^2(n-2)(n-4)}{96}.
\end{array}$$
For even $n\geq 8$ and $n$ mod 4=0,
$$\begin{array}{llll}
\nu_{D_{n}}(E_X)&=\frac{n}{2}\sum^{\frac{n-4}{4}}_{d=1}2d(\frac{n}{2}-1-d)\\
&=\frac{n(n-2)}{2}\frac{\frac{n}{4}\frac{n-4}{4}}{2}-n\frac{\frac{n-4}{4}\frac{n}{4}\frac{n-2}{2}}{6}\\
&=\frac{n^2(n-2)(n-4)}{64}-\frac{n^2(n-2)(n-4)}{192} \\
&=\frac{n^2(n-2)(n-4)}{96}.
\end{array}$$
For even $n\geq 6$,
$$\begin{array}{llll}
\nu_{D_{n}}(E_{XY})&=n\sum^{\frac{n-4}{2}}_{d=1}\sum^{d}_{i=1}(2i-1)\\
&=n\sum^{\frac{n-4}{2}}_{d=1}d^2\\
&=n\frac{\frac{n-4}{2}\frac{n-2}{2}(n-3)}{6}\\
&=\frac{n(n-2)(n-3)(n-4)}{24}.
\end{array}$$
Hence, for even $n\geq 6$,
$$\begin{array}{llll}
\nu(D_{n})&=2\nu_{D_{n}}(E_X)+\nu_{D_{n}}(E_{XY})\\
&=\frac{n^2(n-2)(n-4)}{48}+\frac{n(n-2)(n-3)(n-4)}{24}\\
&=\frac{n(n-2)^2(n-4)}{16}\\
&=\lfloor\frac{n}{2}\rfloor\lfloor\frac{n-1}{2}\rfloor\lfloor\frac{n-2}{2}\rfloor\lfloor\frac{n-3}{2}\rfloor.
\end{array}$$
So we have the following Lemma \ref{lemma: D4}:
\begin{lemma}\label{lemma: D4}
For even $n\geq 6$, $\nu(D_{n})= \lfloor\frac{n}{2}\rfloor\lfloor\frac{n-1}{2}\rfloor\lfloor\frac{n-2}{2}\rfloor\lfloor\frac{n-3}{2}\rfloor$.
\end{lemma}

For odd $n\geq 5$, let
$$\begin{array}{llll}
X&=\{a_{2i},b_{2i+1} \ |\ 0\leq i<(n-1)/2\},\\
Y&=\{a_{2i+1},b_{2i} \ |\ 0\leq i<(n-1)/2\},\\
Z&=\{a_{n-1},b_{n-1} \},\\
E_X&=\{uv \ |\ u,v\in X\},\\
E_{XY}&=\{uv \ |\ u\in X \mbox{ and } v\in Y\},\\
E_Y&=\{uv \ |\ u,v\in Y\},\\
E_{ZXY}&=\{uv \ |\ \ u\in Z \mbox{ and } v\in X\bigcup Y\}.
\end{array}$$
Then
$$\begin{array}{llll}
V(K_{n,n}-nK_2)&=X\bigcup Y\bigcup Z,\ \  X\bigcap Y=X\bigcap Z=Y\bigcap Z=\emptyset.\\
E(K_{n,n}-nK_2)&=E_X\bigcup E_{XY}\bigcup E_{Y}\bigcup E_{ZXY}.
\end{array}$$

In drawings $D_{n}$ with  odd $n\geq 5$, vertices
$a_{2\lceil\frac{n-1}{4}\rceil-1},b_{2\lceil\frac{n-1}{4}\rceil},\cdots,a_{2\lceil\frac{n-1}{4}\rceil-3},b_{2\lceil\frac{n-1}{4}\rceil-2}$
($a_{0}$, $b_{1},\cdots$, $a_{n-3},b_{n-2})$ of $X$ $(Y)$ are placed
equidistantly on the perimeter of the top (bottom) disk, vertex
$a_{n-1}$ $(b_{n-1})$ is placed on the center of the top (bottom)
disk, edges of $E_X$ $(E_Y)$ are drawn on the top (bottom) disk, and
edges of $E_{XY}$ are drawn by shortest helical curves on the
cylinder. For $0\leq i\leq (n-3)/2$, edges $a_{n-1}b_{2i}$
($b_{n-1}a_{2i}$) are drawn on the top(bottom) disk, and edges
$a_{n-1}b_{2i+1}$ ($b_{n-1}a_{2i+1}$) are drawn on the top (bottom)
disk and cylinder. It follows $\nu_{D_{n}}(E_{X}\bigcup E_{Y}\bigcup
E_{XY})=\nu(D_{n-1})$, $\nu_{D_{n}}(E_{ZXY})=0$ and
$$\begin{array}{llll}
&&\nu_{D_{n}}(E_{X}\bigcup E_{Y}\bigcup E_{XY},E_{ZXY})\\
&=&(n-1)(2\lfloor\frac{n-1}{4}\rfloor\lfloor\frac{n-3}{4}\rfloor+\lfloor\frac{n-1}{4}\rfloor+\sum^{\frac{n-5}{2}}_{d=0}d)\\
&=&(n-1)(\lfloor\frac{n-1}{4}\rfloor(2\lfloor\frac{n-3}{4}\rfloor+1)+\frac{\frac{n-5}{2}\frac{n-3}{2}}{2})\\
&=&(n-1)(\frac{(n-1)(n-3)}{8}+\frac{(n-3)(n-5)}{8})\\
&=&\frac{(n-1)(n-3)^2}{4}.
\end{array}$$
By Lemma \ref{lemma: D4},
$$\begin{array}{llll}
\nu(D_{n})&=\nu_{D_{n}}(E_{X}\bigcup E_{Y}\bigcup E_{XY})+\nu_{D_{n}}(E_{X}\bigcup E_{Y}\bigcup E_{XY},E_{ZXY})+\nu_{D_{n}}(E_{ZXY})\\
&=\frac{(n-1)(n-3)^2(n-5)}{16}+\frac{(n-1)(n-3)^2}{4}\\
&=\frac{(n-1)^2(n-3)^2}{16}\\
&=\lfloor\frac{n}{2}\rfloor\lfloor\frac{n-1}{2}\rfloor\lfloor\frac{n-2}{2}\rfloor\lfloor\frac{n-3}{2}\rfloor.
\end{array}$$
So we have the following Lemma \ref{lemma: D5}:
\begin{lemma}\label{lemma: D5}
For odd $n\geq 5$, $\nu(D_{n})= \lfloor\frac{n}{2}\rfloor\lfloor\frac{n-1}{2}\rfloor\lfloor\frac{n-2}{2}\rfloor\lfloor\frac{n-3}{2}\rfloor$.
\end{lemma}

By Lemma \ref{lemma: D1}-\ref{lemma: D5}, we have
\begin{theorem}\label{theorem: knn} For $n\geq 1$,
$cr(K_{n,n}-nK_2)\leq
\lfloor\frac{n}{2}\rfloor\lfloor\frac{n-1}{2}\rfloor\lfloor\frac{n-2}{2}\rfloor\lfloor\frac{n-3}{2}\rfloor$.
\end{theorem}

Since $K_n\times P_2\cong K_{n,n}-nK_2$, we have
\begin{corollary} For $n\geq 1$,
$cr(K_n\times P_2)\leq
\lfloor\frac{n}{2}\rfloor\lfloor\frac{n-1}{2}\rfloor\lfloor\frac{n-2}{2}\rfloor\lfloor\frac{n-3}{2}\rfloor$.
\end{corollary}

\section{Upper bound of $cr(K_{n}\times P_3)$}

\indent \indent In this section, firstly we introduce some topological tools, and then we give a drawing for $K_{n}\times P_3$ based on $D_n$.

We define a structure $M^2_{l,r}$ in the real plane $\mathbb{R}^2$.
For the vertical points $(0,0)$ and $(0,1)$, let
$S_l=\{(l^0_{i},l^1_{i}):0\leq i\leq l-1\}$
($S_r=\{(r^0_{i},r^1_{i}):0\leq i\leq r-1\}$) be a set of
non-vertical bunches of lines in the left (right) semi-plane, such
that the point $(0,j)$ belongs to $l^j_{i}$ ($r^j_{i}$) for $0\leq
i\leq l-1$ ($0\leq i\leq r-1$) and $j=0,1$. In Figure \ref{fig: M212
and M222} we show the drawings of $M^2_{1,2}$ and $M^2_{2,2}$ as
examples.
\begin{figure}[h]
\centering\includegraphics[scale=1]{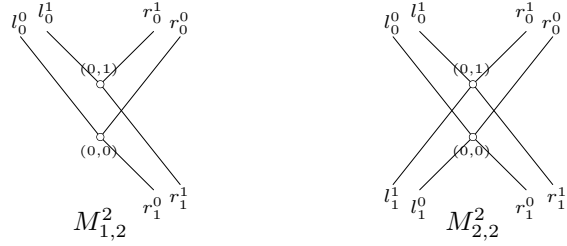}
\caption{\small{Drawings of $M^2_{1,2}$ and $M^2_{2,2}$}}\label{fig:
M212 and M222}
\end{figure}

By counting the crossings in $M^2_{l,r}$, we have
\begin{lemma}\label{lemma: mesh2}
For positive integers $l$ and $r$,
$$\nu_{M^2_{l,r}}(S_l)+\nu_{M^2_{l,r}}(S_r)=(^{l}_{2})+(^{r}_{2})=\frac{l(l-1)}{2}+\frac{r(r-1)}{2}.$$
\end{lemma}

Let
$$\begin{array}{llll}
V(K_{n}\times P_3)&=\{0i,1i,2i\ |\ 0\leq i\leq n-1\},\\
E(K_{n}\times P_3)&=\{(0i,1j),(2i,1j)\ |\ 0\leq i\neq j\leq n-1\}.
\end{array}$$
We construct a drawing $D'_n$ for $K_{n}\times P_3$ based on $D_n$
with $M^2_{\frac{n-1}{2},\frac{n-1}{2}}$ for odd $n$ and
$M^2_{\frac{n}{2},\frac{n-2}{2}}$ or
$M^2_{\frac{n-2}{2},\frac{n}{2}}$ for even $n$. Vertex $a_i$ ($b_i$)
($0\leq i\leq n-1$) in its ``small'' neighborhood in the drawing
$D_n$ is replaced by  two vertices $0i$ and $2i$ vertically (one
vertex $1i$). Now every drawn edge $e$ in $D_n$ which starts from
$a_i$ is replaced by a bunch of two edges described above, which is
drawn along the original edge $e$. Drawings $D'_{n}$ for $n\leq 9$
are shown in Figure \ref{fig: KnP3}.

\begin{figure}[h]
\centering
\includegraphics[scale=0.99]{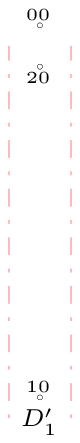}
\hspace{8bp}
\includegraphics[scale=0.99]{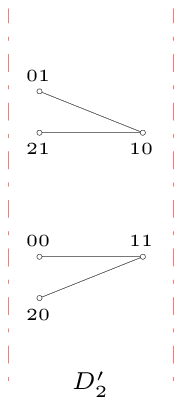}
\hspace{8bp}
\includegraphics[scale=0.99]{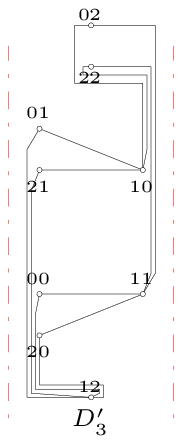}
\hspace{8bp}
\includegraphics[scale=0.99]{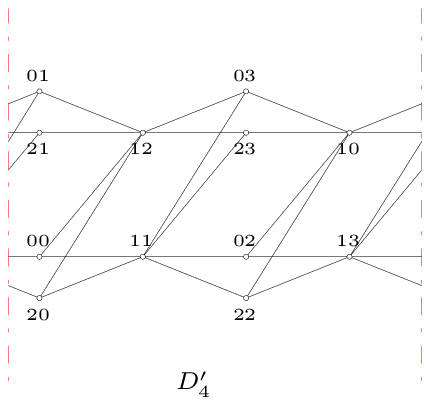}
\hspace{8bp}
\includegraphics[scale=0.99]{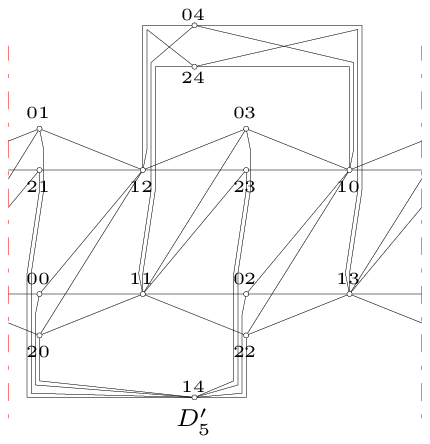}
\includegraphics[scale=0.99]{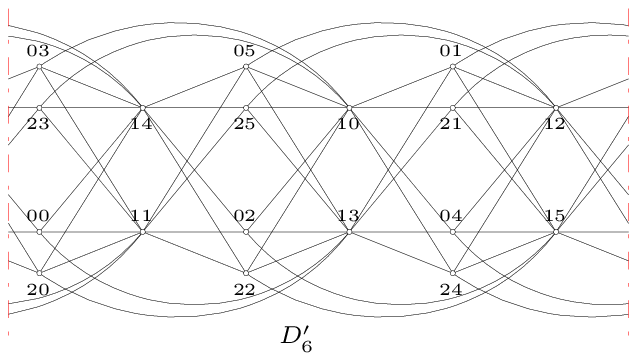}
\includegraphics[scale=0.99]{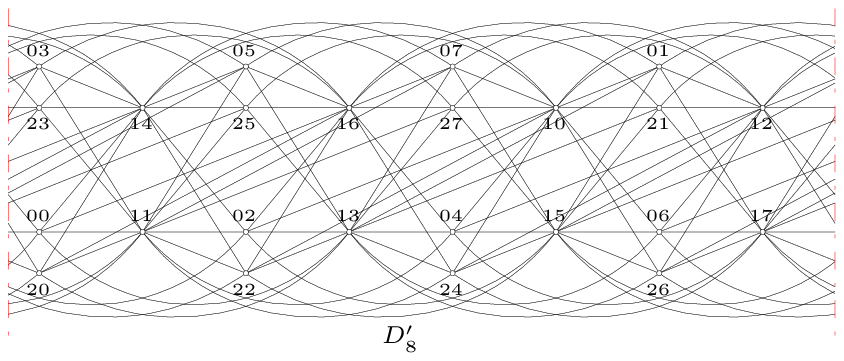}
\includegraphics[scale=0.99]{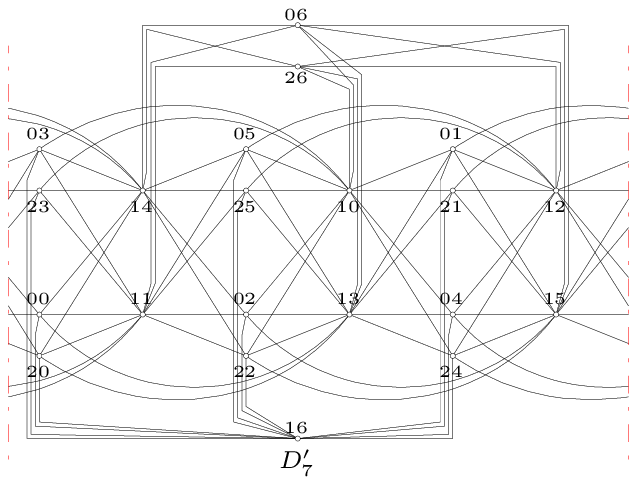}
\includegraphics[scale=0.99]{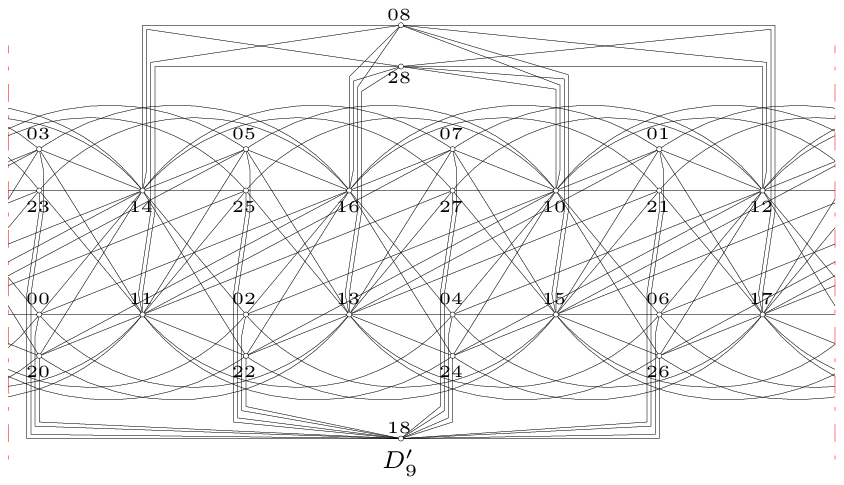}
\caption{\small{Drawings $D'_n$ for $n\leq9$}}\label{fig: KnP3}
\end{figure}

There are totally two types of crossings in $D'_{n}$.

\noindent\textbf{Type 1. } \textbf{Crossings produced by the
crossings in $D_n$}.

Each crossing in $D_n$ will be replaced by $4$ crossings in
$D'_{n+1}$, since we replace every edge with a bunch of two edges.

\noindent\textbf{Type 2. } \textbf{Crossings produced by replacing a
vertex of $D_{n}$ by two vertices in $D'_{n}$}.

The total number of crossings of Type 2 are $n\cdot \nu(M^2_{\frac{n-1}{2},\frac{n-1}{2}})$ for odd $n$ and $n\cdot \nu(M^2_{\frac{n}{2},\frac{n-2}{2}})$ for even $n$.

By Lemma \ref{lemma: mesh2}, for odd $n$,
$$\begin{array}{llll}
\nu(D'_n)&=4\nu(D_n)+n\cdot\nu(M^2_{\frac{n-1}{2},\frac{n-1}{2}})\\
&=4\lfloor\frac{n}{2}\rfloor\lfloor\frac{n-1}{2}\rfloor\lfloor\frac{n-2}{2}\rfloor\lfloor\frac{n-3}{2}\rfloor+n(\frac{\frac{n-1}{2}\frac{n-3}{2}}{2}+\frac{\frac{n-1}{2}\frac{n-3}{2}}{2})\\
&=4\lfloor\frac{n}{2}\rfloor\lfloor\frac{n-1}{2}\rfloor\lfloor\frac{n-2}{2}\rfloor\lfloor\frac{n-3}{2}\rfloor+n\lfloor\frac{n-1}{2}\rfloor\lfloor\frac{n-2}{2}\rfloor.
\end{array}$$
For even $n$,
$$\begin{array}{llll}
\nu(D'_n)&=4\nu(D_n)+n\cdot\nu(M^2_{\frac{n}{2},\frac{n-2}{2}})\\
&=4\lfloor\frac{n}{2}\rfloor\lfloor\frac{n-1}{2}\rfloor\lfloor\frac{n-2}{2}\rfloor\lfloor\frac{n-3}{2}\rfloor+n(\frac{\frac{n}{2}\frac{n-2}{2}}{2}+\frac{\frac{n-2}{2}\frac{n-4}{2}}{2})\\
&=4\lfloor\frac{n}{2}\rfloor\lfloor\frac{n-1}{2}\rfloor\lfloor\frac{n-2}{2}\rfloor\lfloor\frac{n-3}{2}\rfloor+n\lfloor\frac{n-1}{2}\rfloor\lfloor\frac{n-2}{2}\rfloor.
\end{array}$$
Hence we have
\begin{theorem}\label{theorem: knp3} For $n\geq 1$,
$cr(K_{n}\times P_3)\leq 4\lfloor\frac{n}{2}\rfloor\lfloor\frac{n-1}{2}\rfloor\lfloor\frac{n-2}{2}\rfloor\lfloor\frac{n-3}{2}\rfloor+n\lfloor\frac{n-1}{2}\rfloor\lfloor\frac{n-2}{2}\rfloor$.
\end{theorem}

\section{Upper bound of $cr(K_{n}\times C_4)$}

\indent \indent We define another structure $M^4_{l,r}$ in the real
plane $\mathbb{R}^2$. For the vertical points $(0,0)$ and $(0,1)$,
let $S_l=\{(l^0_{i,0},l^0_{i,1},l^1_{i,0},l^1_{i,1}):0\leq i\leq
l-1\}$ ($S_r=\{(r^0_{i,0},r^0_{i,1},r^1_{i,0},r^1_{i,1}):0\leq i\leq
r-1\}$) be a set of non-vertical bunches of lines in the left
(right) semi-plane, such that the point $(0,j)$ belongs to
$l^j_{i,k}$ ($r^j_{i,k}$) for $0\leq i\leq l-1$ ($0\leq i\leq r-1$),
$j=0,1$ and $k=0,1$. In each bunch of lines of $S_l$ ($S_r$), there
is exact one `self' crossing crossed by $l^0_{i,1}$ and $l^1_{i,1}$
($r^0_{i,1}$ and $r^1_{i,1}$). In Figure \ref{fig: M412 and M422} we
show the drawings of $M^4_{1,2}$ and $M^4_{2,2}$ as examples.
\begin{figure}[h]
\centering\includegraphics[scale=1]{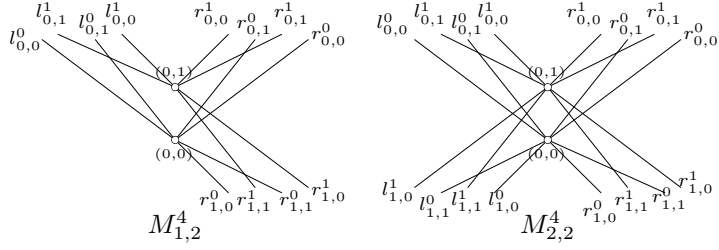}
\caption{\small{Drawings of $M^4_{1,2}$ and $M^4_{2,2}$}}\label{fig:
M412 and M422}
\end{figure}

By counting the crossings in $M^4_{l,r}$, we have
\begin{lemma}\label{lemma: each mesh}
For positive integers $l$ and $r$,
$$\nu_{M^4_{l,r}}(S_l)+\nu_{M^4_{l,r}}(S_r)=(^{2l}_{2})+(^{2r}_{2})=l(2l-1)+r(2r-1).$$
\end{lemma}

Let
$$\begin{array}{llll}
V(K_{n}\times C_4)&=\{0i,1i,2i,3i\ |\ 0\leq i\leq n-1\},\\
E(K_{n}\times C_4)&=\{(0i,1j),(1i,2j),(2i,3j),(3i,0j)\ |\ 0\leq i\neq j\leq n-1\}.
\end{array}$$
We construct a drawing $D''_n$ for $K_{n}\times C_4$ based on $D_n$
with $M^4_{\frac{n-1}{2},\frac{n-1}{2}}$ for odd $n$ and
$M^4_{\frac{n}{2},\frac{n-2}{2}}$ or
$M^4_{\frac{n-2}{2},\frac{n}{2}}$ for even $n$. Vertex $a_i$ ($b_i$)
($0\leq i\leq n-1$) in its ``small'' neighborhood in the drawing
$D_n$ is replaced by  two vertices $0i$ and $2i$ ($1i$ and $3i$)
vertically. Now every drawn edge $e$ in $D_n$ which starts from
$a_i$ ($b_i$) is replaced by a bunch of four edges described above,
which is drawn along the original edge $e$. For even $n\geq 4$, by
drawing edge $l^{0}_{i,1}$ (or
$l^{1}_{i,1},r^{0}_{i,1},r^{1}_{i,1}$) carefully, $2n$ `self'
crossings can be reduced. Drawings $D''_{n}$ for $n\leq 9$ are shown
in Figure \ref{fig: KnC4}.

\begin{figure}[h]
\centering
\includegraphics[scale=0.99]{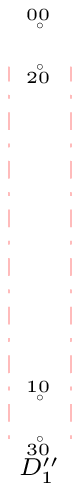}
\hspace{8bp}
\includegraphics[scale=0.99]{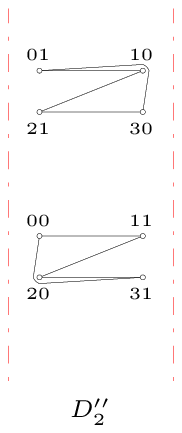}
\hspace{8bp}
\includegraphics[scale=0.99]{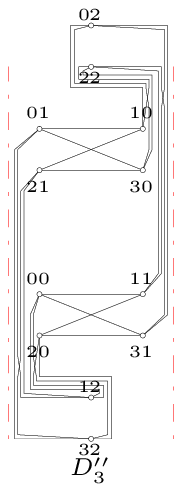}
\hspace{8bp}
\includegraphics[scale=0.99]{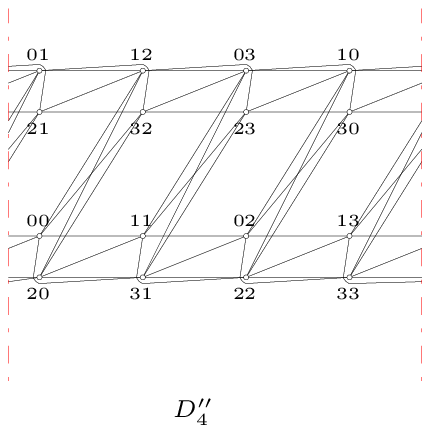}
\hspace{8bp}
\includegraphics[scale=0.99]{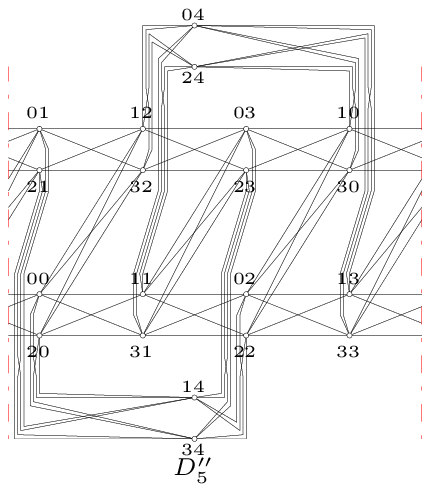}

\includegraphics[scale=0.99]{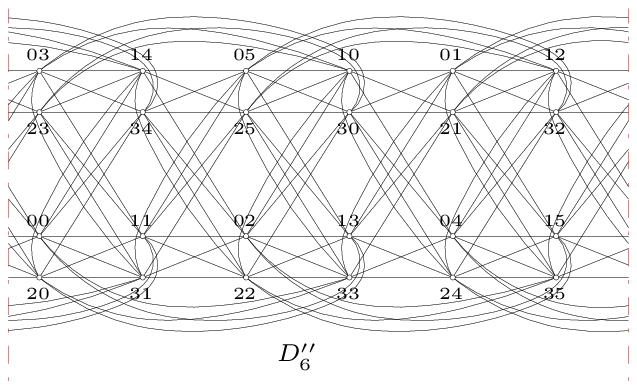}
\includegraphics[scale=0.99]{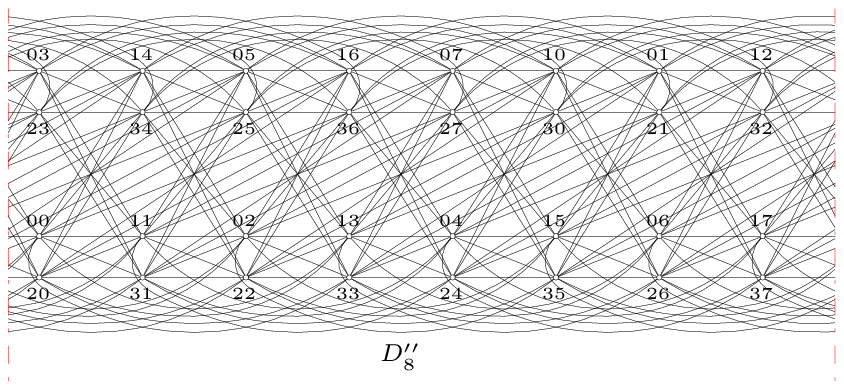}
\includegraphics[scale=0.99]{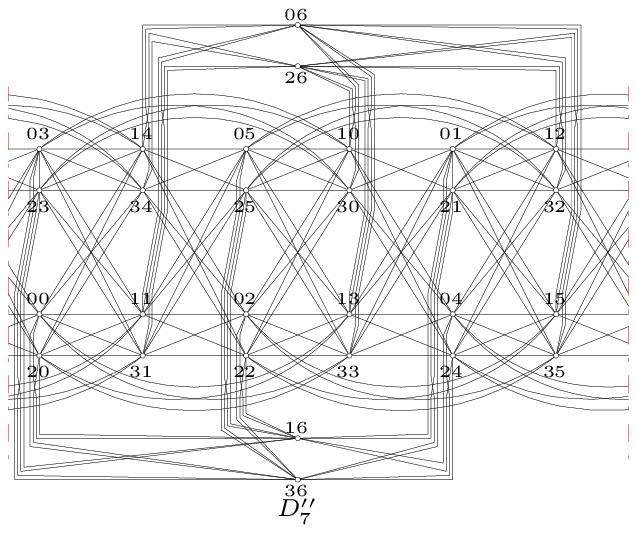}
\includegraphics[scale=0.99]{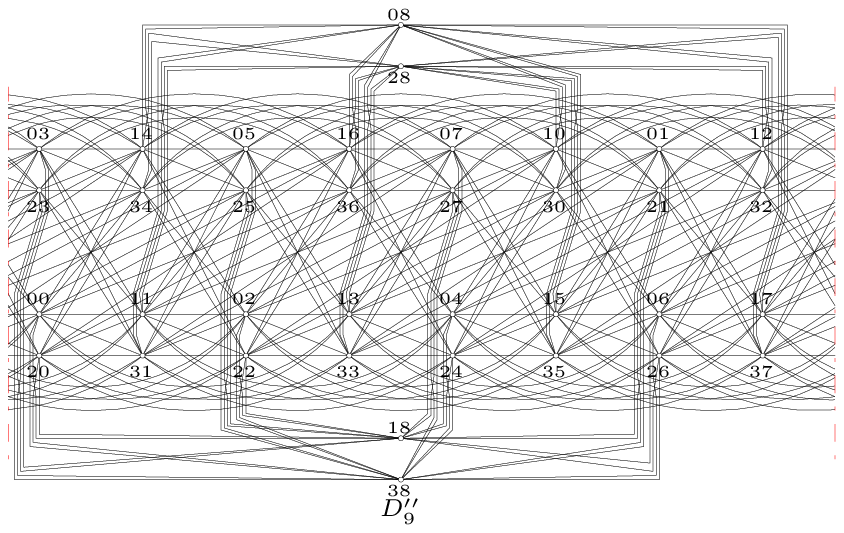}
\caption{\small{Drawings $D''_n$ for $n\leq9$}}\label{fig: KnC4}
\end{figure}

Since $K_1\times C_4\cong4K_1$ and $K_2\times C_4\cong2C_4$, we have
$cr(K_1\times C_4)=cr(K_2\times C_4)=0$. Now, we consider the cases
of $n\geq 3$.

There are totally two types of crossings in $D''_{n}$.

\noindent\textbf{Type 1. } \textbf{Crossings produced by the crossings in $D_n$}.

Each crossing in $D_n$ will be replaced by $16$ crossings in
$D''_{n+1}$, since we replace every edge with a bunch of four edges.

\noindent\textbf{Type 2. } \textbf{Crossings produced by replacing a
vertex of $D_{n}$ by two vertices in $D''_{n}$}.

The total number of crossings of Type 2 are $2n\cdot \nu(M^4_{\frac{n-1}{2},\frac{n-1}{2}})-n(n-1)$ for odd $n$ and $2n\cdot \nu(M^4_{\frac{n}{2},\frac{n-2}{2}})-n(n-1)-2n$ for even $n$. Notice
that we have to subtract $|E(D_n)|=n(n-1)$ since the `self' crossings of each bunch are counted twice (namely both of the endpoints).

By Lemma \ref{lemma: each mesh}, for odd $n\geq3$,
$$\begin{array}{llll}
\nu(D''_n)&=16\nu(D_n)+2n\cdot\nu(M^4_{\frac{n-1}{2},\frac{n-1}{2}})-n(n-1)\\
&=16\lfloor\frac{n}{2}\rfloor\lfloor\frac{n-1}{2}\rfloor\lfloor\frac{n-2}{2}\rfloor\lfloor\frac{n-3}{2}\rfloor+2n(n-1)(n-2)-n(n-1)\\
&=16\lfloor\frac{n}{2}\rfloor\lfloor\frac{n-1}{2}\rfloor\lfloor\frac{n-2}{2}\rfloor\lfloor\frac{n-3}{2}\rfloor+n(n-1)(2n-5).
\end{array}$$
For even $n\geq4$,
$$\begin{array}{llll}
\nu(D''_n)&=16\nu(D_n)+2n\cdot\nu(M^4_{\frac{n}{2},\frac{n-2}{2}})-n(n-1)-2n\\
&=16\lfloor\frac{n}{2}\rfloor\lfloor\frac{n-1}{2}\rfloor\lfloor\frac{n-2}{2}\rfloor\lfloor\frac{n-3}{2}\rfloor+2n(n^2-3n+3)-n(n-1)-2n\\
&=16\lfloor\frac{n}{2}\rfloor\lfloor\frac{n-1}{2}\rfloor\lfloor\frac{n-2}{2}\rfloor\lfloor\frac{n-3}{2}\rfloor+n(n-1)(2n-5).
\end{array}$$
Hence we have
\begin{theorem}\label{theorem: knc4} For $n\geq 3$,
$cr(K_{n}\times C_4)\leq 16\lfloor\frac{n}{2}\rfloor\lfloor\frac{n-1}{2}\rfloor\lfloor\frac{n-2}{2}\rfloor\lfloor\frac{n-3}{2}\rfloor+n(n-1)(2n-5)$.
\end{theorem}

\section{Lower bounds of $cr(K_{n,n}-nK_2)$, $cr(K_{n}\times P_2)$, $cr(K_{n}\times P_3)$ and $cr(K_{n}\times C_4)$}

\indent \indent We shall introduce the lower bound method proposed
by Leighton \cite{L84}. Let $G_1=(V_1,E_1)$ and $G_2=(V_2,E_2)$ be
graphs. An embedding of $G_1$ in $G_2$ is a couple of mappings
$(\varphi,\kappa)$ satisfying
$$\varphi: V_1\rightarrow V_2 \mbox{ is an injection}$$
$$\kappa: E_1\rightarrow \{\mbox{set of all paths in $G_2$}\},$$ such that
if $uv\in E_1$ then $\kappa(uv)$ is a path between $\varphi(u)$ and
$\varphi(v)$. For any $e\in E_2$ define
$$cg_e(\varphi,\kappa)=|\{f\in E_1:e\in \kappa(f)\}|$$
and
$$cg(\varphi,\kappa)=\max\limits_{e\in E_2}\{cg_e(\varphi,\kappa)\}.$$
The value $cg(\varphi,\kappa)$ is called congestion.

\begin{lemma}\label{Lemma congestion} \mbox{\textup{\cite{L84}}} Let $(\varphi,\kappa)$ be an embedding of $G_1$ in
$G_2$ with congestion $cg(\varphi,\kappa)$. Let $\Delta(G_2)$ denote the maximal degree of $G_2$. Then
$$cr(G_2)\geq \frac{cr(G_1)}{cg^2(\varphi,\kappa)}-\frac{|V_2|}{2}\Delta^2(G_2).$$
\end{lemma}

Let $K^x_{m,n}$ be the complete bipartite multigraph of $m+n$
vertices, in which every two vertices are joined by x parallel
edges.

According to De Klerk \cite{D07} and Kainen \cite{K72}, the
following lemmas hold.

\begin{lemma}\label{Lemma crossing of Kmn}\mbox{\textup{\cite{D07}}} $cr(K_{m,n})\geq
0.8594\lfloor\frac{m}{2}\rfloor\lfloor\frac{m-1}{2}\rfloor
\lfloor\frac{n}{2}\rfloor\lfloor\frac{n-1}{2}\rfloor $.
\end{lemma}

\begin{lemma}\label{Lemma crossing of
xKmn}\mbox{\textup{\cite{K72}}}
$cr(K^x_{m,n})=x^2cr(K_{m,n})$.
\end{lemma}

Now we are in a position to show the lower bound of $cr(K_{n,n}-nK_2)$.

\begin{theorem}\label{Theorem Lower Bound for Knn}
$cr(K_{n,n}-nK_2)\geq
\frac{0.8594}{(1+\frac{3}{n-1})^2}\lfloor\frac{n}{2}\rfloor
^2\lfloor\frac{n-1}{2}\rfloor ^2-n(n-1)^2$.
\end{theorem}
\begin{proof} By Lemmas \ref{Lemma congestion},
\ref{Lemma crossing of Kmn} and \ref{Lemma crossing of xKmn}, we
only need to construct an embedding $(\varphi,\kappa)$ of
$K^{(n-1)(n-2)}_{n,n}$ into $K_{n,n}-nK_2$ with congestion
$cg(\varphi,\kappa)=(n-2)(n+2)$.

Let $\alpha^k_{i}\beta^k_{i}$ be the $k$-th $(n-1,2)$-arrangement, where $\alpha^k_{i},\beta^k_{i}\in \{0,1,2,\cdots,n-1\}-\{i\}$ and $\alpha^k_{i}\neq \beta^k_{i}$. Let
$$\begin{array}{llll}
V(K^{(n-1)(n-2)}_{n,n})&=\{u_i,v_i\ |\ 0\leq i\leq n-1\},\\
E(K^{(n-1)(n-2)}_{n,n})&=\{(u_i,v_j)^k\ |\ 0\leq i,j\leq n-1, 1\leq k\leq (n-1)(n-2)\}.
\end{array}$$
Let $\varphi (u_i)=a_i$, $\varphi (v_i)=b_i$,
$\kappa((u_i,v_i)^k)=P_{a_ib_{\alpha^k_{i}}a_{\beta^k_{i}}b_i}$ for
$0\leq i\leq n-1$, and $\kappa((u_i,v_j)^k)=(a_i,b_j)$ for $0\leq
i\neq j\leq n-1$. Then $cg_e(\varphi,\kappa)=(n-2)(n+2)$ for every
$e\in E(K_{n,n}^{(n-1)(n-2)})$. This completes the proof of Theorem
\ref{Theorem Lower Bound for Knn}.
\end{proof}

Since $K_n\times P_2\cong K_{n,n}-nK_2$, we have
\begin{corollary}
$cr(K_n\times P_2)\geq \frac{0.8594}{(1+\frac{3}{n-1})^2}\lfloor\frac{n}{2}\rfloor ^2\lfloor\frac{n-1}{2}\rfloor ^2-n(n-1)^2$.
\end{corollary}

Similar to the proof of Theorem \ref{Theorem Lower Bound for Knn},
by constructing an embedding $(\varphi,\kappa)$ of
$K^{(n-1)(n-2)}_{2n,n}$ into $K_{n}\times P_3$ with congestion
$cg(\varphi,\kappa)=(n-2)(n+2)$, we can get
\begin{theorem}\label{Theorem Lower Bound for Kn2n}
$cr(K_{n}\times P_3)\geq \frac{0.8594}{(1+\frac{3}{n-1})^2}n\lfloor\frac{2n-1}{2}\rfloor\lfloor\frac{n}{2}\rfloor\lfloor\frac{n-1}{2}\rfloor-\frac{3}{2}n(2n-2)^2$.
\end{theorem}

And by constructing an embedding $(\varphi,\kappa)$ of
$K^{(n-1)(n-2)}_{2n,2n}$ into $K_{n}\times C_4$ with congestion
$cg(\varphi,\kappa)=(n-2)(n+2)$, we can get
\begin{theorem}\label{Theorem Lower Bound for K2n2n}
$cr(K_{n}\times P_4)\geq \frac{0.8594}{(1+\frac{3}{n-1})^2}n^2\lfloor\frac{2n-1}{2}\rfloor^2-2n(2n-2)^2$.
\end{theorem}

\section{Conclusion}

By Theorem \ref{theorem: knn}, Theorem \ref{theorem: knp3}, Theorem \ref{theorem: knc4} and Theorem \ref{Theorem Lower Bound for Knn}-\ref{Theorem Lower Bound for K2n2n}, we have {\tiny
 $$\frac{0.8594}{(1+\frac{3}{n-1})^2}\lfloor\frac{n}{2}\rfloor ^2\lfloor\frac{n-1}{2}\rfloor ^2-n(n-1)^2\leq cr(K_n\times P_2)=cr(K_{n,n}-nK_2)\leq
 \lfloor\frac{n}{2}\rfloor\lfloor\frac{n-1}{2}\rfloor\lfloor\frac{n-2}{2}\rfloor\lfloor\frac{n-3}{2}\rfloor,$$
 $$\frac{0.8594}{(1+\frac{3}{n-1})^2}n\lfloor\frac{2n-1}{2}\rfloor\lfloor\frac{n}{2}\rfloor\lfloor\frac{n-1}{2}\rfloor-\frac{3}{2}n(2n-2)^2\leq cr(K_{n}\times P_3)\leq
4\lfloor\frac{n}{2}\rfloor\lfloor\frac{n-1}{2}\rfloor\lfloor\frac{n-2}{2}\rfloor\lfloor\frac{n-3}{2}\rfloor+n\lfloor\frac{n-1}{2}\rfloor\lfloor\frac{n-2}{2}\rfloor,$$
$$\frac{0.8594}{(1+\frac{3}{n-1})^2}n^2\lfloor\frac{2n-1}{2}\rfloor^2-2n(2n-2)^2\leq cr(K_{n}\times P_4)\leq
16\lfloor\frac{n}{2}\rfloor\lfloor\frac{n-1}{2}\rfloor\lfloor\frac{n-2}{2}\rfloor\lfloor\frac{n-3}{2}\rfloor+n(n-1)(2n-5) \mbox{ for }n\geq 3.$$} From \cite{F86}, we have Lemma \ref{lemma: Di}
\begin{lemma}\label{lemma: Di}
Let $i$ be the least number of the edges of a graph $G$ whose deletion from $G$ results in a planar subgraph of $G$, then $cr(G)\geq i$.
\end{lemma}
By Lemma \ref{lemma: Di}, we have
\begin{lemma}\label{lemma: K55}
$cr(K_{5,5}-5K_2)=4$.
\end{lemma}
\begin{proof} Let $i$ be the least number of the edges of graph $G=K_{5,5}-5K_2$ whose deletion from $G$ results in a planar subgraph $G^{*}$ of $G$. Consider $G^{*}$ has 10 vertices, $20-i$ edges. Let $D^{*}$ be a
planar drawing of $G^{*}$ and $f$ denote the number of the faces in $D^{*}$. Then according to the Euler Polyhedron Formula,
$$\begin{array}{llll}
10-(20-i)+f=2,\\
f=12-i.
\end{array}$$
Since the girth of $G$ is 4, by counting the number of edges of each face in $D^{*}$, we have
$$\begin{array}{llll}
(12-i)\times 4\leq 2(20-i),\\
i\geq 4.
\end{array}$$
By Lemma \ref{lemma: Di}, $cr(K_{5,5}-5K_2)\geq 4$. By Theorem
\ref{theorem: knn}, $cr(K_{5,5}-5K_2)\leq 4$. So
$cr(K_{5,5}-5K_2)=4$.
\end{proof}
 Furthermore, we have the following conjecture.
\begin{conjecture}\label{conjecture: Dn}
$cr(K_{n,n}-nK_2)=\lfloor\frac{n}{2}\rfloor\lfloor\frac{n-1}{2}\rfloor\lfloor\frac{n-2}{2}\rfloor\lfloor\frac{n-3}{2}\rfloor$.
\end{conjecture}
Conjecture \ref{conjecture: Dn} holds for $n\leq 5$ by Lemma
\ref{lemma: D1} and Lemma \ref{lemma: K55}.


\begin{thebibliography}{99}

\bibitem{A04}
J. Adamsson,
\newblock Arrangements, circular arrangements and the crossing number of $C_7\Box C_n$,
\newblock {\it J.Combin. Theory Ser. B} 90 (2004) 21-39.

\bibitem{A96}
D. Archdeacon,
\newblock Topological graph: a survey,
\newblock {\it Cong. numerantium} 115 (1996) 5-54.

\bibitem{BW78}
L. Beineke, R. Welson,
\newblock Selected topics in graph theory,
\newblock {\it Academic Press} (1978) 68-72.

\bibitem{BK64}
Z. Bla$\check{z}$ek, N. Koman,
\newblock A Minimal plobrem concerning complet plane graphs, in M.
Fiedler(Ed.) Theory of Graphs and its Applications, Proceedings of
the Symposium on Smolenice, 1963
\newblock {\it Publ. House of the Czechoslovak Academy of Sciences, Prague} (1964) 118-117.

\bibitem{D07}
E. De Klerk, E. V. Pasechnik, A. Schrijver,
\newblock Reduction of symmetric semidifinite programs using the regular representation,
\newblock {\it Math Prograph, Ser B,} 109 (2007) 613-624.

\bibitem{F86}
S. Fiorini,
\newblock On the crossing number of generalized Petersen graphs,
\newblock {\it Ann, Discrete Math.} 30 (1986) 211-242.

\bibitem{GJ83}
M. R. Garey, D. S. Johnson,
\newblock Crossing number is NP-complete,
\newblock {\it SIAM J. Alg. Disc. Math.} 4 (1983) 312--316.

\bibitem{G69}
R. K. Guy,
\newblock The decline and fall of Zarankiewicz's theorem, in F.
Harary(Ed. ),Proof Techniques in Graph Theory,
\newblock {\it Academic Press, New York, London} (1969) 63--69.

\bibitem{JD12}
P. K. Jha, S. Devisetty,
\newblock Orthogonal drawings and crossing numbers of the Kronecker
product of two cycles
\newblock {\it J. Parallel Distrib. Comput.} 72 (2012) 195-204.

\bibitem{K72}
P. C. Kainen,
\newblock A lower bound for crossing numbers of graphs with applications to $K_n$, $K_{p,q}$, and
$Q(d)$.
\newblock J. Combin. Theory Ser. B {\bf 12}, 287--298 (1972)

\bibitem{K70}
D. J. Kleitman,
\newblock The crossing number of $K_{5,n}$,
\newblock {\it J. Combinatorial Theory} 9 (1970) 315-323.

\bibitem{L84}
F. T. Leighton,
\newblock New lower bound techniques for VLSI,
\newblock {\it Math. Systems Theory} \textbf{17} (1984) 47--70.

\bibitem{W93}
D. R. Woodall,
\newblock Cyclic-order graphs and Zarankiewicz's crossing-number conjecture,
\newblock {\it J. Graph Theory} 17 (1993) 657-671.

\bibitem{Z54}
K. Zarankiewicz,
\newblock On a problem of P.Tur$\acute{a}$n concerling graphs,
\newblock {\it Fund Math} 41 (1954) 137-145.

\bibitem{Z07}
W. Zheng, X. Lin, Y. Yang, C. Cui,
\newblock On the crossing number of $K_m\Box P_n$,
\newblock {\it Graphs and Combinatorics} 23 (2007) 327-336.

\bibitem{Z08}
W. Zheng, X. Lin, Y. Yang,
\newblock The crossing number of $K_{2,m}\Box P_n$,
\newblock {\it Discrete Mathematics} 308 (2008) 6639-6644.

\bibitem{Z081}
W. Zheng, X. Lin, Y. Yang, C. Deng,
\newblock On the crossing number of $K_{m}\Box C_n$ and $K_{m,l}\Box P_n$,
\newblock {\it Discrete Applied Mathematics} 156 (2008) 1892-1907.


\end{thebibliography}
\end{document}